\documentclass[12pt]{article}
\usepackage[small,bf]{caption}
\usepackage{fullpage,graphicx,psfrag,url,ar,rotating,url,hyperref}
\usepackage{amsmath,amssymb,amsthm,enumitem,subcaption,appendix}
\usepackage{booktabs,adjustbox}

\newcommand{\BEAS}{\begin{eqnarray*}}
\newcommand{\EEAS}{\end{eqnarray*}}
\newcommand{\BEQ}{\begin{equation}}
\newcommand{\EEQ}{\end{equation}}
\newcommand{\BIT}{\begin{itemize}}
\newcommand{\EIT}{\end{itemize}}
\newcommand{\BMAT}{\begin{bmatrix}}
\newcommand{\EMAT}{\end{bmatrix}}

\newcommand{\eg}{{\it e.g.}}
\newcommand{\ie}{{\it i.e.}}

\newcommand{\ones}{\mathbf 1}
\newcommand{\reals}{{\mbox{\bf R}}}

\newtheorem{theorem}{Theorem}

\newtheorem{proposition}[theorem]{Proposition}
\theoremstyle{definition}

\newcounter{oursection}

\newcommand{\adarad}{AdaRad}

{

\newcounter{algorithmctr}[section]
\renewcommand{\thealgorithmctr}{\thesection.\arabic{algorithmctr}}
\newenvironment{algdesc}%
{\refstepcounter{algorithmctr}\begin{list}{}{%
			\setlength{\rightmargin}{0\linewidth}%
			\setlength{\leftmargin}{.05\linewidth}}%
		\rmfamily\small
		\item[]{\setlength{\parskip}{0ex}\hrulefill\par%
			\nopagebreak{\bfseries\textsf{Algorithm \thealgorithmctr~}}}}%
	{{\setlength{\parskip}{-1ex}\nopagebreak\par\hrulefill} \end{list}}

\title{Operator Splitting for Adaptive Radiation Therapy with Nonlinear 
	Health Dynamics}
\author{Anqi Fu \and Lei Xing \and Stephen Boyd}

\begin{document}
\maketitle

\begin{abstract}
We present an optimization-based approach to radiation treatment planning 
over time. Our approach formulates treatment planning as an optimal control 
problem with nonlinear patient health dynamics derived from the standard 
linear-quadratic cell survival model. As the formulation is nonconvex, we 
propose a method for obtaining an approximate solution by solving a sequence 
of convex optimization problems. This method is fast, efficient, and robust 
to model error, adapting readily to changes in the patient's health between 
treatment sessions. Moreover, we show that it can be combined with the operator 
splitting method ADMM to produce an algorithm that is highly scalable and 
can handle large clinical cases. We introduce an open-source Python 
implementation of our algorithm, \adarad{}, and demonstrate its performance 
on several examples.
\end{abstract}

\section{Introduction}
In radiation therapy, beams of ionizing radiation are transmitted into a 
patient, damaging both tumor cells and normal tissue. The goal of radiation 
treatment planning is to deliver enough dose to the tumor so that diseased 
cells are killed, while avoiding excessive injury to the normal tissue and 
organs-at-risk (OARs). This is achieved by optimizing the beam intensity 
profile, or fluence map, subject to constraints on the dose to certain 
parts of the patient's anatomy. The fluence map optimization problem is 
well-studied \cite{RomeijnAhuja:2003,RomeijnLi:2004,AlemanGlaser:2010,Gao:2016}, 
and technology like intensity-modulated radiation therapy (IMRT) is now 
widespread in the clinic \cite{ZelefskyFuks:2000,WoldenChen:2006,GuptaAgarwal:2012,Webb:2010}. 

Treatment in practice usually takes place over multiple sessions. A 
clinician will divide up the total prescribed dose into smaller dose 
fractions, which are delivered over the course of several weeks or months. 
This permits normal tissue time to recover and repair cell damage, 
but also gives tumors an opportunity to proliferate, especially when the 
treatment course is long. A study of 4,338 prostate cancer patients showed 
that biochemical failure increases by 6\% for every 1 week increase in 
treatment time, with a dose equivalent of proliferation of 0.24 Gy/day 
\cite{ThamesKuban:2010}. Thus, an important question in treatment planning 
is how to choose the sequence of deliverable doses such that they balance 
these temporal effects on a patient's health.

\subsection{Related work}
\label{sec:lit_review}
Early clinical practitioners split the prescribed dose equally over a 
fixed number of sessions. While convenient, this method does not account 
for errors or uncertainty in the treatment process. For example, due to 
movements of the patient's anatomy, the expected dose may differ from the 
actual dose to an anatomical structure. If the actual dose is observable, 
a common way to compensate for this is to divide the residual dose (\ie, 
the difference between the prescribed and cumulative actual dose) across 
the remaining sessions. This then becomes the new per-session dose goal. 
In \cite{ZerdaXing:2007}, the authors solve for the beam intensities by 
minimizing the sum-of-squared difference between this goal dose and the 
expected dose. They compare the results when errors are perfectly known, 
so the expected dose is equal to the actual dose, with the results when 
errors are assumed to be zero. A similar approach is taken in 
\cite{FerrisVoelker:2004}, except the errors are modeled explicitly as a 
random shift in the surrounding voxels. Instead of the dose to each voxel, 
\cite{SirPollock:2012} work with the equivalent uniform dose (EUD), a 
value that captures the biological effect of a dose distribution over a 
region. Their objective is to minimize the sum of the EUD over all 
treatment criteria subject to bounds on the EUD of the tumor and normal 
tissues. To solve this problem, they employ methods from approximate 
dynamic programming coupled with a discrete probabilistic model of the 
dose error.

The papers we have discussed so far only focus on the dose to the patient. 
By contrast, \cite{KimPhillips:2009} introduces a Markov decision process 
model that includes both the dose (action) and the patient's health state. 
Each choice of dose induces a transition to a particular health state with 
some probability. Making this idea concrete, \cite{MizutaShirato:2012} 
define the health of a tumor (resp. OAR) to be the radiation (resp. damage) 
effect of the delivered dose, as calculated from the linear-quadratic (LQ) 
model of cell survival \cite{Fowler:1989}. They analyze a simple example 
with one tumor and one OAR and find that the optimal fractionation scheme 
is either a single session delivery of the full dose or equal dose 
fractions, depending on the relationship between the LQ parameters. The 
authors of \cite{BortfeldRTU:2015} extend this analysis to incorporate 
accelerated tumor repopulation and show that the dose per session increases 
over the treatment course. Using simulated annealing, \cite{YangXing:2005} 
solve a similar treatment planning problem based on the LQR model, which 
captures all 4 Rs (repair of sublethal damage, repopulation, redistribution, 
and reoxygenation) of cellular radiation response \cite{BrennerHlatky:1995}.

In \cite{ChanMisic:2013}, the authors take a probabilistic approach to 
patient dynamics. They model the patient's breathing motion using a 
probability mass function (PMF) over a finite set of states. They then 
solve a robust optimization problem that enforces dose bounds over a set 
of PMFs, which represent uncertainty during treatment. This uncertainty 
set is updated after the dose fraction is delivered, and the problem is 
re-solved for the next session. A follow-up paper \cite{MisicChan:2015} 
numerically studies the effects of adjusting the target dose based on the 
dose delivered to date. It shows that this type of adjustment can lead to 
a high degree of heterogeneity in the per-fraction dose distribution, 
which is undesirable from a medical standpoint. The paper recommends an 
alternative method of uncertainty set adaptation to mitigate these effects.

The above analyses provide insight into the tradeoffs between different 
fractionation schemes in a simple setting. However, most clinical cases are 
more complex, involving multiple tumors, OARs, and nonlinear constraints. 
For instance, dose-volume (\ie, percentile) constraints are widely used to 
limit the radiation exposure of a percentage of an anatomical structure, 
such as the spine. These constraints are nonconvex, but can be approximated 
by a convex restriction \cite{HalabiCraft:2006,ZST+Zinchenko:2013,FuUngunXing:2019}. 
In \cite{SaberianGhateKim:2016}, the authors consider a dynamic setting with 
multiple OARs and dose-volume constraints. Starting from a given set of 
beam intensities, they solve for the optimal number of sessions and 
OAR sparing factors. They also derive sufficient conditions under which the 
optimal treatment consists of equal dose fractions. In a follow-up paper 
\cite{SaberianGhateKim:2017}, the authors integrate the spatial and 
temporal aspects of the problem, treating both beam intensities and number 
of sessions as variables. Restricting their attention to equal fractions, 
they propose a two-stage solution algorithm: in the first stage, they solve 
for the optimal beams given each potential fixed number of sessions, and in 
the second stage, they select the number of sessions based on the optimal 
objectives from the first stage. They show that their method achieves 
better tumor ablation than conventional IMRT or the spatiotemporally 
separated method.

Perhaps the paper most similar to ours is \cite{KimPhillips:2012}. In it, 
the authors propose a stochastic control formulation of the adaptive 
treatment planning problem with multiple tumors and OARs. They estimate 
the radiation response of the tumors with a log-linear cell kill model and 
the response of the OARs with the standard LQ model. Their goal is to 
minimize the expected number of tumor cells at the end of treatment subject 
to bounds on the radiobiological impact on the OARs. Uncertainty arises in 
the cell model parameters, which may fluctuate randomly between sessions, 
representing unpredictable changes in the patient's health status. The 
authors fix the number of sessions and focus on optimizing with respect to 
the beam intensities. They show that their problem is convex, so can be 
solved using a combination of standard stochastic control methods and 
off-the-shelf convex solvers, and provide several examples demonstrating 
the effectiveness of their approach.

\subsection{Contribution}
\label{sec:contrib}
In this paper, we integrate the stochastic control approach with a 
distributed optimization algorithm to produce a method for efficient 
large-scale adaptive treatment planning. As clinical cases are quite 
complex, with tens of thousands of beams and treatment that takes place 
over months, such methods are necessary to construct plans in a timely 
fashion. (See \cite{JiaJiang:2014} for a review of previous work on 
high-performance computing in radiation therapy, particularly treatment 
optimization). We formulate the adaptive treatment planning problem as a 
finite-horizon nonconvex optimal control problem. To solve it, we introduce 
an operator splitting algorithm, which is based on solving a sequence of 
convex approximations. Our algorithm is naturally parallelizable and can 
handle a large number of beams, sessions, and anatomical targets or OARs. 
Moreover, it can be combined with model predictive control to produce 
treatment plans that are robust to errors and uncertainty about the 
patient's health status. We illustrate our algorithm's performance on a 
synthetic case, as well as a large prostate cancer case, and provide an 
implementation in the Python package \adarad{}: \url{https://github.com/anqif/adarad}.

\section{Problem formulation}
\label{sec:prob}
In radiation treatment, beams of ionizing radiation are delivered to a 
patient from an external source. The goal is to damage or kill diseased 
tissue, while minimizing harm to surrounding healthy organs. A course 
of treatment is generally divided into $T$ sessions. At the start of 
session $t$, the clinician chooses the intensity levels of the $n$ 
beams, denoted by $b_t \in \reals_+^n$. Typically, $T \approx 20$ and 
$n$ is on the order of $10^3$ to $10^4$. We are interested in 
determining the best sequence of beam intensities 
$b = (b_1,\ldots,b_T)$, otherwise known as a {\em treatment plan}, 
subject to upper bounds $B_t \in \bar\reals_+^n$ on $b_t$ for 
$t=1,\ldots,T$.

\paragraph{Anatomy and doses.} The beams irradiate an area containing 
$K$ anatomical structures, labeled $i \in \{1,\ldots,K\}$, where usually 
$K < 10$.  A subset $\mathcal{T} \subset \{1,\ldots,K\}$ are targets/tumors 
and the rest are OARs. The dose delivered to each structure is linear in 
the beam intensities. We write the dose vector $d_t = A_tb_t$ with $A_t \in 
\reals_+^{K \times n}$ a known matrix that characterizes the physical 
effects and define $d = (d_1,\ldots,d_T)$. Notice that since $b_t$ and 
$A_t$ are nonnegative, $d_t \geq 0$.

In every session, we impose a penalty on $d_t$ 
via a {\em dose penalty function} $\phi_t: \reals^K \rightarrow \reals 
\cup \{\infty\}$. A common choice is
\[
	\phi_t(d_t) = \theta_t^Td_t + \xi_t^Td_t^{\circ 2},
\]
where $\theta_t \in \reals^K$ and $\xi_t \in \reals_+^K$ are constants. Here 
$d_t^{\circ 2} = d_t \odot d_t$ denotes the vector $d_t$ with each element 
squared.   
The total dose penalty over all sessions is
\[
	\phi(d) = \sum_{t=1}^T \phi_t(d_t).
\]
Additionally, we enforce upper bound constraints $d_t \leq D_t$, where 
$D_t \in \bar \reals_+^K$ is the maximum dose in session $t$.

\paragraph{Health dynamics.} To assess treatment progress, we examine the 
health status of each anatomical structure and encode these statuses in a 
vector $h_t \in \reals^K$. For now, the details of this encoding do not 
matter. Typically, $h_{ti}$ represents an estimate of the total surviving 
cells in structure $i$. Hence if $i \in \mathcal{T}$, a smaller $h_{t,i}$ is 
desirable (since the tumor is shrinking), while if $i \notin \mathcal{T}$, 
a larger $h_{t,i}$ is desirable.

From an initial $h_0$, the health status evolves in response to the 
radiation dose and various other biophysical factors  that depend on 
the patient's anatomy, generating a health trajectory 
$h = (h_1,\ldots,h_T)$. Here we represent its dynamics as
\BEQ
\label{eq:health_dynamics}
	h_t = f_t(h_{t-1}, d_t), \quad t = 1,\ldots,T,
\EEQ
where $f_t:\reals^K \times \reals^K \rightarrow \reals^K$ is a known 
mapping function. In this paper, we focus on the linear-quadratic (LQ) 
model in which
\BEQ
\label{eq:health_dyn_quad}
	f_{t,i}(h_{t-1}, d_t) = h_{t-1,i} - \alpha_{t,i}d_{t,i} - 
		\beta_{t,i}d_{t,i}^2 + \gamma_{t,i}, \quad i = 1,\ldots,K, 
	\quad t = 1,\ldots,T
\EEQ
with constants $\alpha_t \in \reals^K, \beta_t \in \reals_+^K$, and 
$\gamma_t \in \reals^K$. This model is commonly used to approximate 
cellular response to radiation \cite{Fowler:1989,ThamesHendry:1987,Brenner:2008}. 
Specifically, in the LQ + time framework \cite{TravisTucker:1987}, $h_{t,i}$ 
is the log of the fraction of surviving cells in structure $i$ after a dose 
$d_{t,i}$, while $\alpha_{t,i}/\beta_{t,i}$ and $\gamma_{t,i}$ are constants 
related to the structure's survival curve and repair/repopulation rate, 
respectively. Notice that equation (\ref{eq:health_dyn_quad}) implies that 
the health status of each structure evolves independently of the others.

\paragraph{Health penalty and constraints.} In order to control the 
patient's health, we introduce a {\em health penalty function} $\psi_t: 
\reals^K \rightarrow \reals \cup \{\infty\}$ that imposes a penalty on 
$h_t$. Moreover, we assume that
\BEQ
\label{cond:health_mono}
\psi_t(h_t) = \psi_t(h_{t,1},\ldots,h_{t,K}) \mbox{ is monotonically }
\begin{cases}
	\mbox{increasing in } h_{t,i} & i \in \mathcal{T} \\
	\mbox{decreasing in } h_{t,i} & i \notin \mathcal{T}
\end{cases}
\EEQ
for $t = 1,\ldots,T$. This means that for a target, the health penalty 
increases as the health status increases, while for an organ-at-risk, 
the health penalty decreases as the health status increases. The 
assumption is reasonable if, for instance, $h_t$ is a measure of cell 
survival in session $t$, so a lower (higher) status is desirable for a 
target (organ-at-risk). An example of a penalty function that satisfies 
(\ref{cond:health_mono}) is
\[
	\psi_t(h_t) = \overline w^T(h_t - h_t^{\text{goal}})_+ + 
		\underline w^T(h_t - h_t^{\text{goal}})_-,
\]
where $h_t^{\text{goal}} \in \reals^K$ is the desired health status and 
$\underline w \in \reals_+^K$ and $\bar w \in \reals_+^K$ are parameters 
with $\underline w_i = 0$ for $i \in \mathcal{T}$ and $\overline w_i = 0$ 
for $i \notin \mathcal{T}$. Here we define $(x)_+ = \max(x, 0)$ and 
$(x)_- = -\min(x, 0)$ applied elementwise to $x$. The total health penalty 
is
\[
	\psi(h) = \sum_{t=1}^T \psi_t(h_t).
\]
In addition, we enforce bounds $H_t \in \bar \reals^K$ on the health 
status such that $h_{t,i} \leq H_{t,i}$ for $i \in \mathcal{T}$ and 
$h_{t,i} \geq H_{t,i}$ for $i \notin \mathcal{T}$.

\paragraph{Optimal control problem.} 
Given an initial health status $h_0$, we wish to select a treatment plan 
that minimizes the total penalty across all sessions. Thus, our problem 
is
\BEQ
\label{prob:dyn_single}
\begin{array}{lll}
	\mbox{minimize} & \sum_{t=1}^T \phi_t(d_t) + 
		\sum_{t=1}^T \psi_t(h_t) & \\
	\mbox{subject to} & h_t = f_t(h_{t-1}, d_t), 
		\quad & t = 1,\ldots,T, \\
	& h_{t,i} \leq H_{t,i},~ i \in \mathcal{T}, 
		\quad h_{t,i} \geq H_{t,i},~ i \notin \mathcal{T}, 
		\quad & t = 1,\ldots,T, \\
	& d_t = A_tb_t, \quad 0 \leq d_t \leq D_t, 
 \quad 0 \leq b_t \leq B_t, \quad & t = 1,\ldots,T
\end{array}
\EEQ
with variables $(b_1,\ldots,b_T), (d_1,\ldots,d_T)$, and 
$(h_1,\ldots,h_T)$. This is a discrete-time optimal control problem. If 
$\phi_t$ and $\psi_t$ are convex and $f_t$ is affine, \eg, $f_t$ is given 
by (\ref{eq:health_dyn_quad}) with quadratic dose effect $\beta_t = 0$, it 
is also convex and can be solved directly using standard convex solvers. 

\section{Lossless relaxation}
\label{sec:relaxation}
For the remainder of this paper, we restrict our attention to a convex 
objective function and linear-quadratic health dynamics 
(\ref{eq:health_dyn_quad}). In this case, condition 
(\ref{cond:health_mono}) allows us to relax the health dynamics constraint 
so problem (\ref{prob:dyn_single}) can be written equivalently as
\BEQ
\label{prob:relax}
\begin{array}{lll}
	\mbox{minimize} & \sum_{t=1}^T \phi_t(d_t) + 
		\sum_{t=1}^T \psi_t(h_t) & \\
	\mbox{subject to} & h_{t,i} \geq f_{t,i}(h_{t-1}, d_t),
		~ i \in \mathcal{T}, \quad & t = 1,\ldots,T, \\
	& h_{t,i} \leq f_{t,i}(h_{t-1}, d_t),~ i \notin \mathcal{T}, 
		\quad & t = 1,\ldots,T, \\
	& h_{t,i} \leq H_{t,i},~ i \in \mathcal{T}, \quad h_{t,i} \geq H_{t,i},~ 
		i \notin \mathcal{T}, \quad & t = 1,\ldots,T, \\
	& d_t = A_tb_t, \quad 0 \leq d_t \leq D_t, 
	\quad 0 \leq b_t \leq B_t, \quad & t = 1,\ldots,T.
\end{array}
\EEQ
The equality constraint $h_t = f_t(h_{t-1}, d_t)$ has been replaced with 
two inequality constraints: a lower bound for targets and an upper bound 
for OARs. Notice that the first inequality is the only nonconvex constraint 
in (\ref{prob:relax}). Our relaxed problem has the same solution set as 
(\ref{prob:dyn_single}) because these two inequalities are tight at the 
optimum.

\begin{proposition}
\label{prop:relax_tight}
~
Let $(b^{\star},d^{\star},h^{\star})$ be a solution to problem 
(\ref{prob:relax}). If conditions (\ref{eq:health_dyn_quad}) and 
(\ref{cond:health_mono}) hold, 
\[
	h_t^{\star} = f_t(h_{t-1}^{\star}, d_t^{\star}), \quad t = 1,\ldots,T.
\]
\end{proposition}
\begin{proof}
Suppose there exist some $t \in \{1,\ldots,T\}$ and $i \in \mathcal{T}$ 
such that $h_{t,i}^{\star} > f_{t,i}(h_{t-1}^{\star}, d_t^{\star})$. Then, 
we can choose an $\epsilon > 0$ such that $h_{t,i}^{\star} > h_{t,i}^{\star} - 
\epsilon > f_{t,i}(h_{t-1}^{\star}, d_t^{\star})$. Since 
$f_{s,i}(h_{s-1}, d_s)$ is nondecreasing in $h_{s-1,i}$ for all 
$s \in \{1,\ldots,T\}$, the point $(b^{\star},d^{\star},\hat h)$ with
\[
	\hat h_{s,j} = \begin{cases}
		h_{s,j}^{\star} - \epsilon & s = t,~j = i \\
		h_{s,j}^{\star} & \mbox{otherwise}
	\end{cases}
\]
is feasible for problem (\ref{prob:relax}) because 
\[
	\hat h_{t,i} > f_{t,i}(h_{t-1}^{\star}, d_t^{\star}) \geq f_{t,i}(\hat h_{t-1}, d_t^{\star}), 
	\quad h_{t+1,i}^{\star} \geq f_{t+1,i}(h_t^{\star}, d_t^{\star}) \geq f_{t+1,i}(\hat h_t, d_t^{\star}), 
\]
and $\hat h_{t,i} < h_{t,i}^{\star} \leq H_{t,i}$. Moreover, by condition 
(\ref{cond:health_mono}), $\psi_t(\hat h_t) < \psi_t(h_t^{\star})$ so 
$(b^{\star},d^{\star},\hat h)$ achieves a lower objective value than 
$(b^{\star},d^{\star},h^{\star})$, contradicting our original assumption. 
An analogous argument holds for $t \in \{1,\ldots,T\}$ and $i \notin 
\mathcal{T}$ such that $h_{t,i}^{\star} < f_{t,i}(h_{t-1}^{\star}, d_t^{\star})$ 
with $\hat h_{t,i} = h_{t,i}^{\star} + \epsilon$.
\end{proof}

\section{Sequential convex optimization}
\label{sec:seq_cvx_opt}

\subsection{Algorithm description}
\label{sec:ccp}
Problem (\ref{prob:relax}) is in general nonconvex because the target's health 
dynamics constraint
\BEQ
\label{constr:target_relax}
	h_{t,i} \geq f_{t,i}(h_{t-1}, d_t),~ i \in \mathcal{T}, \quad t = 1,\ldots,T
\EEQ
is nonconvex when any $\beta_t \neq 0$. However, we can derive an estimate 
of its optimum by solving a sequence of convex approximations. Each 
approximation is formed by linearizing the health dynamics function 
(\ref{eq:health_dyn_quad}) around a fixed dose point and replacing the 
right-hand side of (\ref{constr:target_relax}) with this linearization 
minus a slack variable. The slack allows for a degree of error in the 
approximation and is penalized in the objective.

More precisely, let $d_t^s \in \reals^K$ for $t = 1,\ldots,T$. Define 
the linearized dynamics function
\BEQ
\label{eq:dyn_target_lin}
	\hat f_{t,i}(h_{t-1}, d_t; d_t^s) = h_{t-1,i} - \alpha_{t,i}d_{t,i} - 
	\beta_{t,i}d_{t,i}^s(2d_{t,i} - d_{t,i}^s) + \gamma_{t,i},
	\quad i = 1,\ldots,K.
\EEQ
This function is an upper bound on the LQ function (\ref{eq:health_dyn_quad}) 
because $\beta_t \geq 0$. We replace the nonconvex constraint 
(\ref{constr:target_relax}) in problem (\ref{prob:relax}) with the affine 
constraint
\BEQ
\label{constr:target_lin_eq}
	h_{t,i} = \hat f_{t,i}(h_{t-1}, d_t; d_t^s) - \delta_{t,i},~ 
		i \in \mathcal{T}, \quad t = 1,\ldots,T,
\EEQ
where $\delta_t \in \reals_+^K$ is a slack variable. (The inequality 
can been tightened into an equality due to Proposition 
\ref{prop:relax_tight}). Convex approximation $s$ is then
\BEQ
\label{prob:ccp_slack}
\begin{array}{lll}
	\mbox{minimize} & \sum_{t=1}^T \phi_t(d_t) + \sum_{t=1}^T \psi_t(h_t) + 
		\lambda \sum_{t=1}^T \ones^T \delta_t & \\
	\mbox{subject to} & h_{t,i} = \hat f_{t,i}(h_{t-1}, d_t; d_t^s) - \delta_{t,i},
	~ i \in \mathcal{T}, \quad \delta_t \geq 0 \quad & t = 1,\ldots,T, \\
	& h_{t,i} \leq f_{t,i}(h_{t-1}, d_t),~ i \notin \mathcal{T}, 
	\quad & t = 1,\ldots,T, \\
	& h_{t,i} \leq H_{t,i},~ i \in \mathcal{T}, \quad h_{t,i} \geq H_{t,i},~ 
	i \notin \mathcal{T}, \quad & t = 1,\ldots,T, \\
	& d_t = A_tb_t, \quad 0 \leq d_t \leq D_t, 
	\quad 0 \leq b_t \leq B_t, \quad & t = 1,\ldots,T
\end{array}
\EEQ
with variables $(b_1,\ldots,b_T), (d_1,\ldots,d_T), (h_1,\ldots,h_T)$, and 
$(\delta_1,\ldots,\delta_T)$ and slack penalty parameter $\lambda > 0$. The 
parameter $\lambda$ is typically determined empirically or by trial-and-error. 
This problem is convex and can be solved using standard convex solvers. 
Given a solution to (\ref{prob:ccp_slack}), we set the next linearization 
point $d^{s+1} = (d_1^{s+1},\ldots,d_T^{s+1})$ equal to the optimal dose.

\begin{algdesc}
\label{algo:seq_cvx}
	\emph{Sequential convex optimization.}
	\begin{tabbing}
		{\bf input}: initial point $d^0$, parameter	$\lambda > 0$. 
			\\*[\smallskipamount]
		{\bf for} $s = 0,1,\ldots$ {\bf do} \\
		\qquad \= 1.\ \emph{Linearize.} For $t = 1,\ldots,T$, form the 
			linearization (\ref{eq:dyn_target_lin}) around $d_t^s$. \\
		\> 2.\ \emph{Solve.} Set $d^{s+1}$ equal to an optimal dose of 
			problem (\ref{prob:ccp_slack}). \\
		{\bf until} stopping criterion (\ref{cond:seq_cvx_stop}) is satisfied.
	\end{tabbing}
\end{algdesc}

Algorithm \ref{algo:seq_cvx} is a special case of the convex-concave 
procedure (CCP) \cite{YuilleRangarajan:2003,LippBoyd:2016,DCCP}, which is 
itself a form of majorization-minimization 
\cite{HunterLange:2012,SunBabuPalomar:2017}. CCP is a heuristic for 
finding a local optimum of a nonconvex optimization problem. It is 
guaranteed to converge; indeed, when certain differentiability conditions 
are met, it converges to a stationary point \cite{SriperLanck:2009}. 
As a descent algorithm, CCP is usually terminated when the change in 
the objective falls below some user-specified threshold $\epsilon > 0$, 
\ie, 
\BEQ
\label{cond:seq_cvx_stop}
	p_{\text{opt}}^s - p_{\text{opt}}^{s+1} < \epsilon,
\EEQ
where $p_{\text{opt}}^s$ is the optimal objective of problem 
(\ref{prob:ccp_slack}).  
In our simple experiments, we have found that an initial linearization 
point of $d^0 = 0$ and threshold of $\epsilon = 10^{-3}$ produce good 
results.

\subsection{Illustrative example}
\label{sec:ex_simple}
\paragraph{Problem instance.} We consider an example with $n = 1000$ beams 
divided into $50$ bundles of $20$ parallel beams each, positioned evenly 
around a half-circle. There are $K = 5$ structures, a single target 
$\mathcal{T} = \{1\}$ and four OARs (including generic body voxels) 
depicted in Figure \ref{fig:ex1_structs}. Treatment takes place 
over $T = 20$ sessions, so the basic problem has $nT + 2KT = 20200$ 
variables.

\begin{figure}
	\begin{center}
		\includegraphics[height=0.35\textheight]{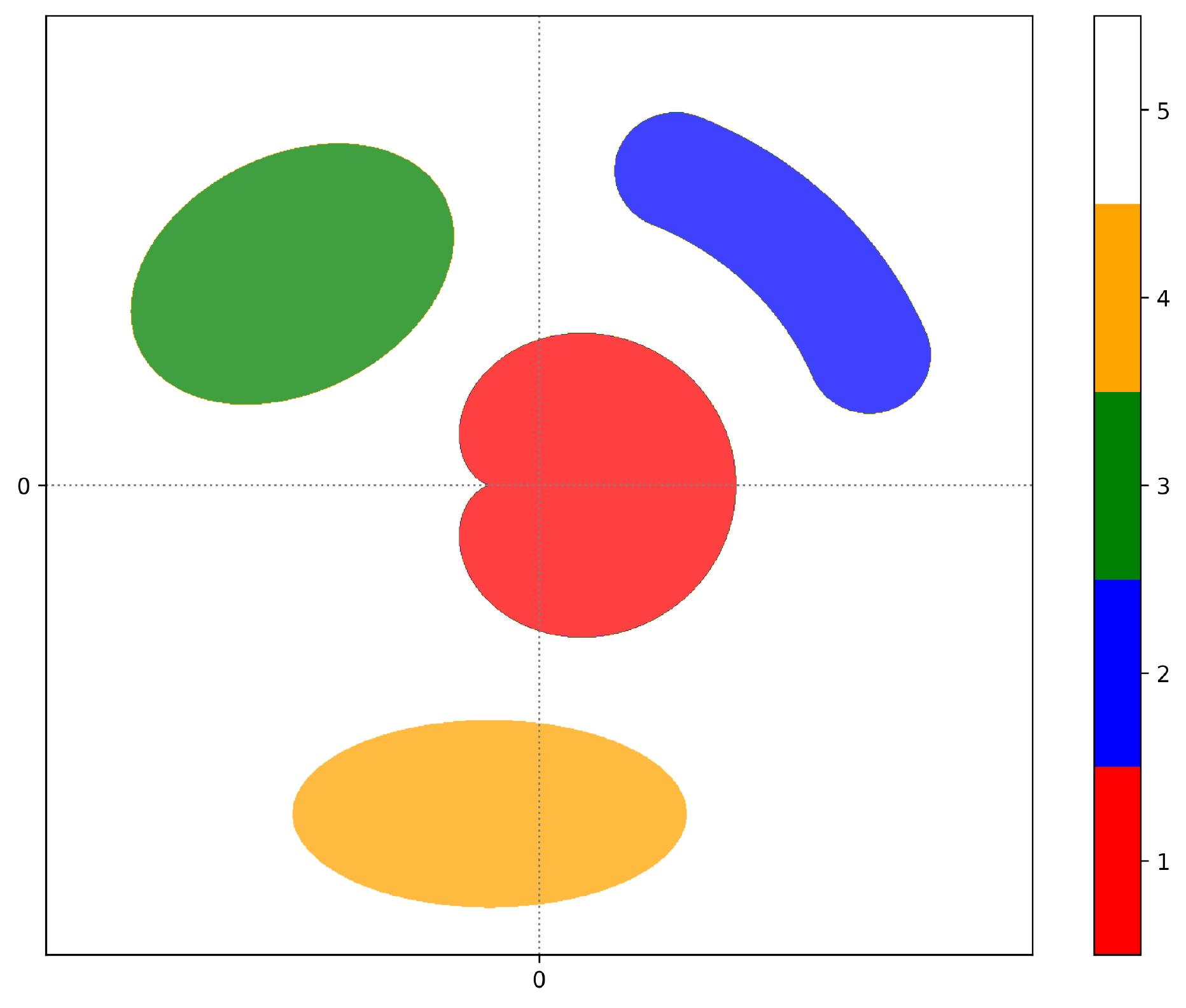}
	\end{center}
	\caption{Anatomical structures for Example \ref{sec:ex_simple}. Red is 
		the target ($i = 1$), while green ($i = 2$), blue ($i = 3$), and 
		orange ($i = 4$) are specific OARs. White denotes the non-target 
		body voxels ($i = 5$).}
	\label{fig:ex1_structs}
\end{figure}

The patient's initial health status is $h_0 = (1,0,0,0,0)$. His status 
evolves according to equation (\ref{eq:health_dyn_quad}) with
\begin{align*}
	\alpha_t &= (0.01, 0.50, 0.25, 0.15, 0.005), \\
	\beta_t  &= (0.001, 0.05, 0.025, 0.015, 0.0005), \\
	\gamma_t &= (0.05, 0, 0, 0, 0)
\end{align*}
over all sessions $t = 1,\ldots,T$. 

We set the health penalty function to
\[
	\psi_t(h_t) = (h_{t,1})_+ + \sum_{i=2}^5 (h_{t,i})_-, 
		\quad t = 1,\ldots,T. 
\]
This function penalizes positive statuses of the target and negative 
statuses of the OARs. Moreover, we constrain the target's health status 
to be $h_{t,1} \leq 2.0$ for $t = 1,\ldots,15$ and $h_{t,1} \leq 0.05$ for the 
remaining sessions, and we enforce a bound on the other structures' health 
statuses of $(h_{t,2}, h_{t,3}, h_{t,4}, h_{t,5}) \geq (-1.0, -2.0, -2.0, -3.0)$. 
Thus,
\[
	H_t = \begin{cases}
		(2.0, -1.0, -2.0, -2.0, -3.0) & t = 1,\ldots,15 \\
		(0.05, -1.0, -2.0, -2.0, -3.0) & t = 16,\ldots,T.
	\end{cases}
\] 

For the dose penalty function, we choose
\[
	\phi_t(d_t) = \sum_{i=1}^4 d_{t,i}^2 + 0.25d_{t,5}^2, \quad t = 1,\ldots,T.
\]
In addition, we restrict the dose and beam intensity to be no more than 
$D_t = 20$ and $B_t = 1.0$, respectively, over all sessions $t$.

\paragraph{Computational details.} We implemented Algorithm 
\ref{algo:seq_cvx} in Python using CVXPY \cite{CVXPY} and solved problem 
(\ref{prob:ccp_slack}) with MOSEK \cite{MOSEK}. From an initial $d^0 = 0$ 
and $\lambda = 10^4$, the algorithm converged in $11$ iterations to a 
threshold of $\epsilon = 10^{-3}$. Total runtime was approximately $17$ 
seconds on a 64-bit Ubuntu OS desktop with $8$ 4-core Intel i7-4790k / $4.00$ 
GHz CPUs and $16$ GB of RAM.

\paragraph{Results and analysis.} The optimal treatment plan is depicted in 
Figure \ref{fig:ex1_beams}. Beams are densely clustered diagonal from the 
vertical, striking the target while largely sparing the OARs. As the 
sessions continue, the number of beams slowly increases, damaging some 
of the less sensitive organs ($i=3$ and $4$). Then at $t=16$, when the 
target's health bound becomes more stringent, the beam density drops 
precipitously so that only a narrow bundle remains focused on the target, 
keeping its health status at the desired level.

Figure \ref{fig:ex1_traj} shows the radiation dose and health status 
resulting from this plan. The latter was computed by plugging the optimal 
dose into equation (\ref{eq:health_dyn_quad}). Total dose to the target 
($i=1$) and body voxels ($i=5$) far exceed the dose to any other 
structures. By the end of treatment, the target's health status has fallen 
to a steady $0.05$, while the health statuses of the OARs remain within 
their respective lower limits.

\begin{figure}
	\begin{center}
		\includegraphics[width=0.95\textwidth]{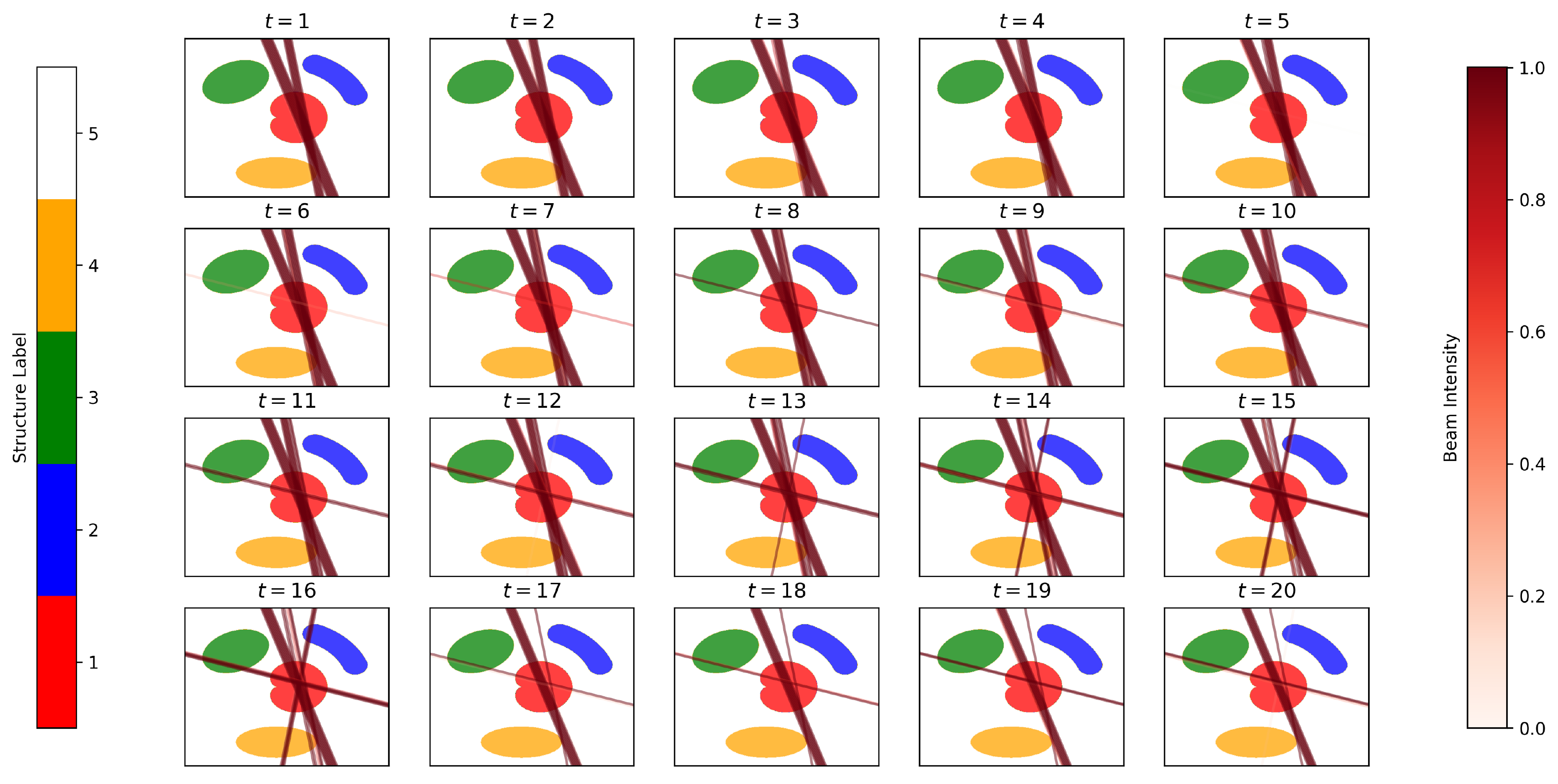}
	\end{center}
	\caption{Optimal beam intensities for Example \ref{sec:ex_simple}.}
	\label{fig:ex1_beams}
\end{figure}

\begin{figure}
	\begin{center}
		\begin{subfigure}[b]{0.95\textwidth}
			\caption{Dose Trajectories}
			\label{fig:ex1_traj_dose}
			\includegraphics[width=\textwidth]{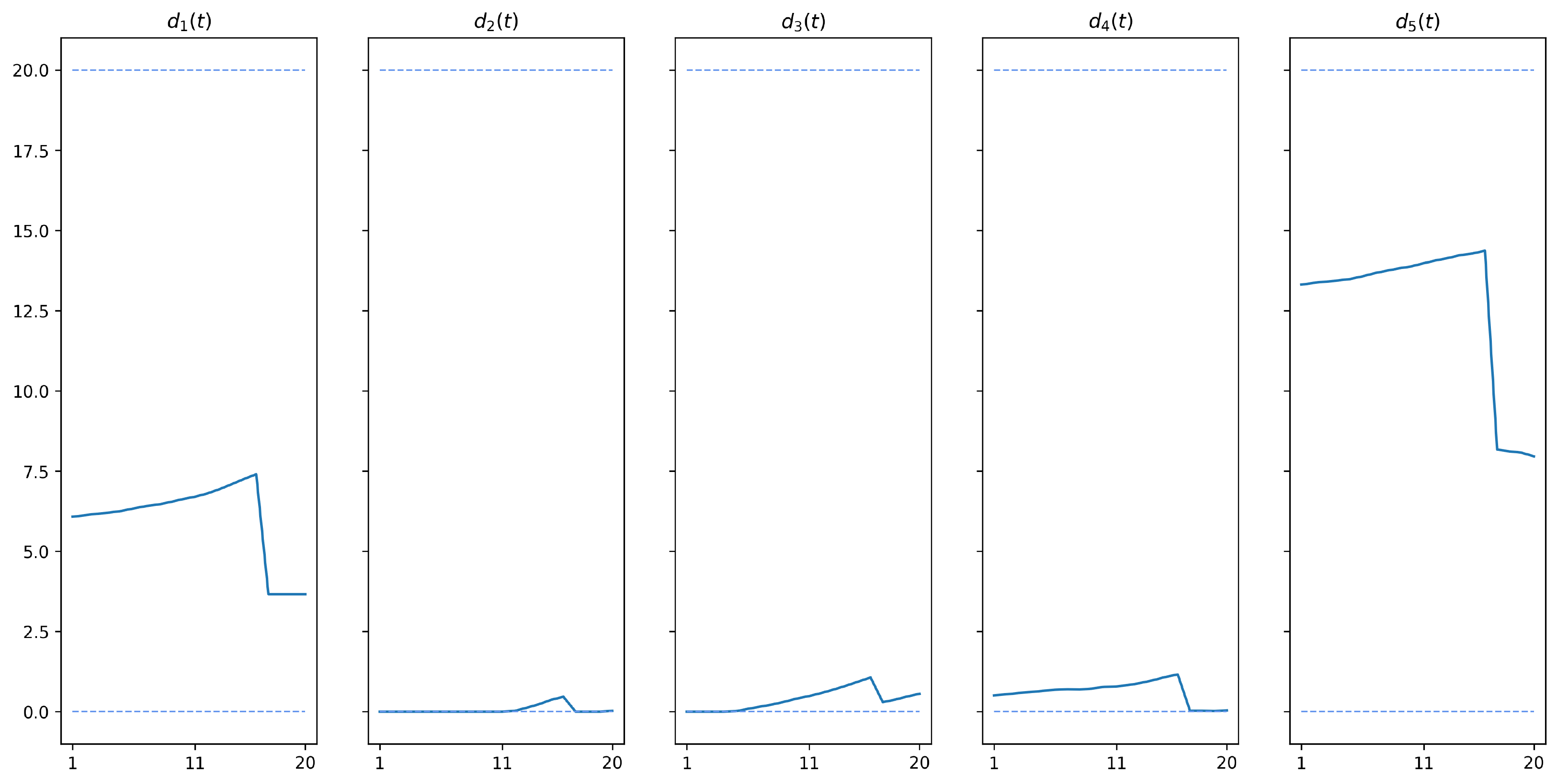}
		\end{subfigure}
		\par\bigskip
		\begin{subfigure}[b]{0.95\textwidth}
			\caption{Health Trajectories}
			\label{fig:ex1_traj_health}
			\includegraphics[width=\textwidth]{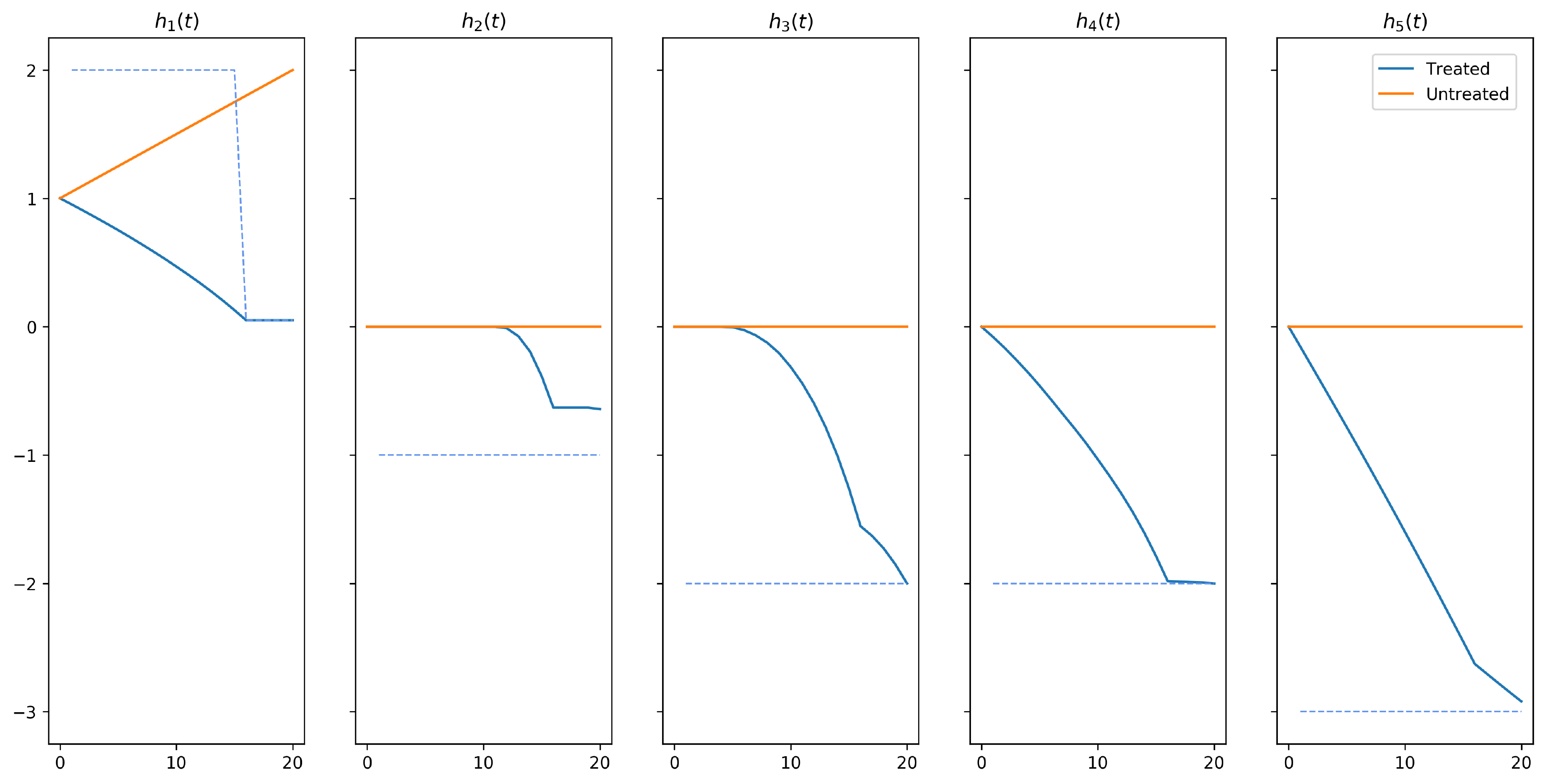}
		\end{subfigure}
		\caption{Optimal (a) radiation dose and (b) health status 
			trajectories for Example \ref{sec:ex_simple}. Dashed lines 
			represent the bounds on the dose and health status ($D_{t,i}$ 
			and $H_{t,i}$, respectively).}
		\label{fig:ex1_traj}
	\end{center}
\end{figure}

\clearpage

\section{Model predictive control}
\label{sec:mpc}

\subsection{Algorithm description}
\label{sec:mpc_desc}
So far, we have assumed that at the time of planning, $f_t$ perfectly 
captures the health dynamics from $t=1,\ldots,T$. This is rarely true in 
practice. A patient's anatomy changes unpredictably between sessions, 
affecting the dispersion of radiation beams and the course of their health 
status. We can incorporate these changes into problem (\ref{prob:dyn_single}) 
using model predictive control (MPC). 

MPC is a powerful technique for automatic control of complex, nonlinear, 
stochastic systems. It performs extremely well even when the dynamics are 
approximated by a simple model, since the system's state is updated 
regularly and new information is incorporated into the solution. This is 
particularly fitting for radiation treatment planning. 

As is customary in MPC, we first convert the state variable constraints 
in the original problem into soft constraints, \ie, we remove the 
inequality constraints on $h$ in (\ref{prob:dyn_single}) and add a 
penalty for violating them to the objective. Let $c_{\tau}: \reals^K 
\rightarrow \reals$ be the corresponding {\em health violation penalty 
function}, defined as
\[
	c_{\tau}(h_{\tau}) = \sum_{i \in \mathcal{T}} (h_{\tau, i} - H_{\tau, i})_+ 
	+ \sum_{i \notin \mathcal{T}} (H_{\tau, i} - h_{\tau, i})_+, 
	\quad \tau = 1,\ldots,T.
\]
This penalty function allows us to accommodate new and unexpected changes in 
the patient's health, such as the metastasis of a tumor that renders it 
impossible to control without exceeding the health damage limit of an OAR.

We are now ready to describe MPC for our model. At the beginning of each session 
$t$, we observe $A_t, f_t$, and the patient's true health status, $h_{t-1}$, then 
form the problem
\BEQ
\label{prob:dyn_single_mpc_slack}
\begin{array}{lll}
	\mbox{minimize} & \sum_{\tau=t}^T \phi_{\tau}(d_{\tau}) + 
	\sum_{\tau=t}^T \psi_{\tau}(h_{\tau}) + 
	\eta \sum_{\tau=t}^T c_{\tau}(h_{\tau}) & \\
	\mbox{subject to} & h_{\tau} = f_t(h_{\tau-1}, d_{\tau}),
	\quad & \tau = t,\ldots,T, \\
	& d_{\tau} = A_t b_{\tau}, 
	\quad 0 \leq d_{\tau} \leq D_{\tau}, 
	\quad 0 \leq b_{\tau} \leq B_{\tau}, 
	\quad & \tau = t,\ldots,T
\end{array}
\EEQ
with variables $(b_t,\ldots,b_T), (d_t,\ldots,d_T)$, and $(h_t,\ldots,h_T)$ 
and violation penalty parameter $\eta > 0$. Since $c_{\tau}$ is convex, 
problem (\ref{prob:dyn_single_mpc_slack}) is convex and can be solved 
using a slight variation on Algorithm \ref{algo:seq_cvx}. Let $\bar b = 
(\bar b_t,\ldots, \bar b_T)$ be the optimal treatment plan. We carry out 
only the first treatment, $\bar b_t$, and update our observations $A_{t+1},
f_{t+1}$, and $h_t$ based on the patient's response. This process repeats 
until all $T$ sessions have been completed.

\subsection{Illustrative example}
\label{sec:ex_simple_mpc}
\paragraph{Problem instance.} We return to the setting of Example 
\ref{sec:ex_simple}, except now, the health dynamics are modeled with some 
error. Specifically, let $h_{t-1}$ be the patient's health status at the 
beginning of session $t$ and $d_t$ the dose delivered during session $t$. 
Our model predicts the status will become $\hat h_t = f_t(h_{t-1},d_t)$. 
In fact, at the beginning of the next session, we observe the true health 
status to be
\[
	(h_t)_i = \begin{cases}
	\max(\hat h_t + \omega_t, 0)_i & i \in \mathcal{T} \\
	\min(\hat h_t + \omega_t, 0)_i & i \notin \mathcal{T}
	\end{cases},
\]
where $\omega_t \in \reals^K$ is drawn from $N(\mu,\sigma^2I)$. This 
random process continues for $t = 1,\ldots,T$. 

For this example, we choose $\mu = 0$ and $\sigma = 0.1$. The rest of the 
functions and parameter values are identical to \ref{sec:ex_simple}. In 
particular, we still employ the LQ model (\ref{eq:health_dyn_quad}) with 
constant $\alpha_t, \beta_t$, and $\gamma_t$ even though the health status 
is now stochastic. We plan the treatment using MPC with $\eta = 10^4$ and 
compare the results to those generated by the naive approach, which simply 
solves problem (\ref{prob:dyn_single}) once prior to session 1.

\paragraph{Computational details.} We solved problem 
(\ref{prob:dyn_single_mpc_slack}) using Algorithm \ref{algo:seq_cvx} with 
$\lambda = 10^4$ and $\epsilon = 10^{-3}$. For the initial dose in session 
1, we chose $d^0 = 0$. In each subsequent session $t$, we set $d^0$ to be 
the (truncated) optimal dose point from the previous session, 
$(d_t^{\star},\ldots,d_T^{\star})$. With these parameters, the algorithm 
took an average of 7 iterations per session to achieve convergence; most 
runs completed in only 3--4 iterations. The total runtime was 116 seconds.

\paragraph{Results and analysis.} Figure \ref{fig:ex2_mpc_beams} depicts 
the treatment plan output by MPC. Most beams are aimed slightly diagonal 
from the vertical, similar to the naive plan (Figure \ref{fig:ex1_beams}) 
up to session 14. Then, the bundles of beams start to grow sparser and fan 
out, hitting more areas of the OARs. This sparse irradiation pattern 
continues until the final session, when there is a brief spike in intensity 
to bring the target's health status into the desired range.

In Figure \ref{fig:ex2_mpc_traj_dose}, we plot the dose trajectories of 
the MPC plan (green) and the naive plan (blue). The MPC curves are more 
jagged with a large spike at the end of treatment. However, in each 
structure, the area under the MPC and naive dose curves remains on par. 
Thus, we conclude that the MPC plan delivers about the same amount of 
radiation as the naive plan, only spread across a wider range of beam 
angles/intensities so as to compensate for uncertainty in the health 
dynamics model.

This strategy results in better patient health as shown in Figure 
\ref{fig:ex2_mpc_traj_health}. The MPC plan reduces the target's health 
status to $0.05$, while maintaining the health status of the OARs at a 
high level. Indeed, the health of these organs under the MPC plan exceeds 
their health under the naive plan by a significant margin in all but 
structure 4, where the two are relatively equal up until the last session.

\begin{figure}
	\begin{center}
		\includegraphics[width=0.95\textwidth]{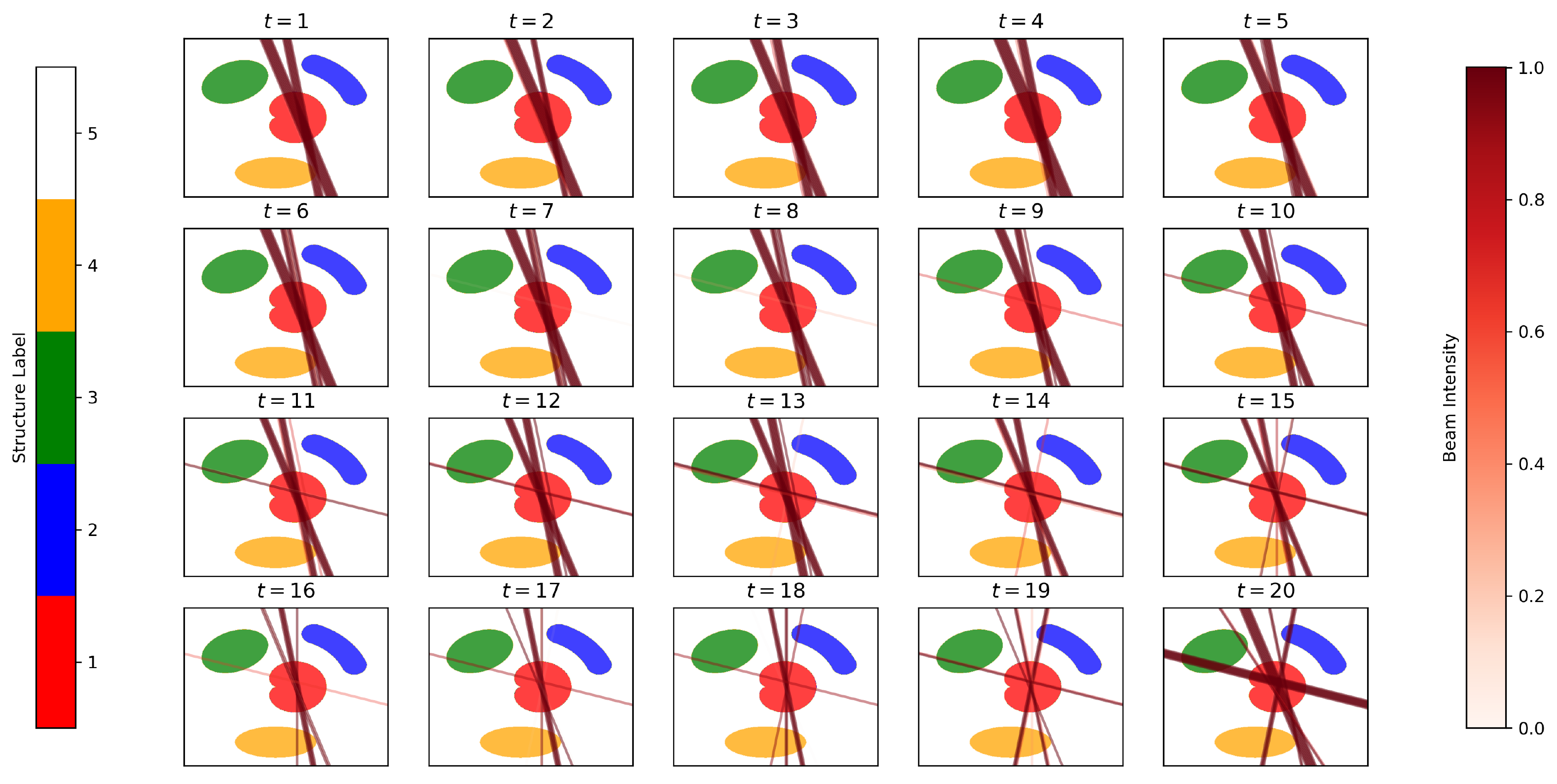}
	\end{center}
	\caption{Optimal beam intensities for Example \ref{sec:ex_simple_mpc} using MPC.}
	\label{fig:ex2_mpc_beams}
\end{figure}

\begin{figure}
	\begin{center}
		\begin{subfigure}[b]{0.95\textwidth}
			\caption{Dose Trajectories}
			\label{fig:ex2_mpc_traj_dose}
			\includegraphics[width=\textwidth]{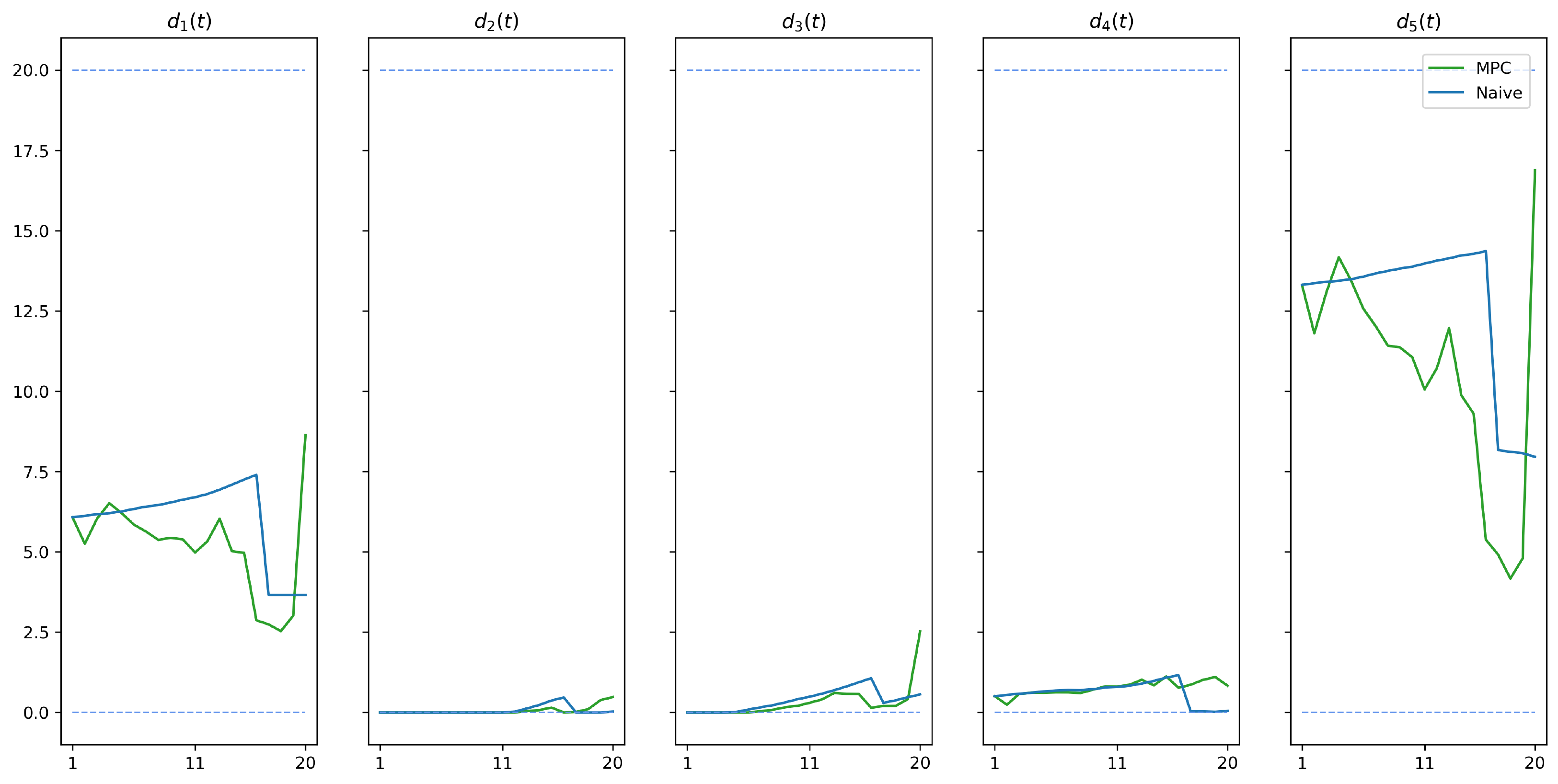}
		\end{subfigure}
		\par\bigskip
		\begin{subfigure}[b]{0.95\textwidth}
			\caption{Health Trajectories}
			\label{fig:ex2_mpc_traj_health}
			\includegraphics[width=\textwidth]{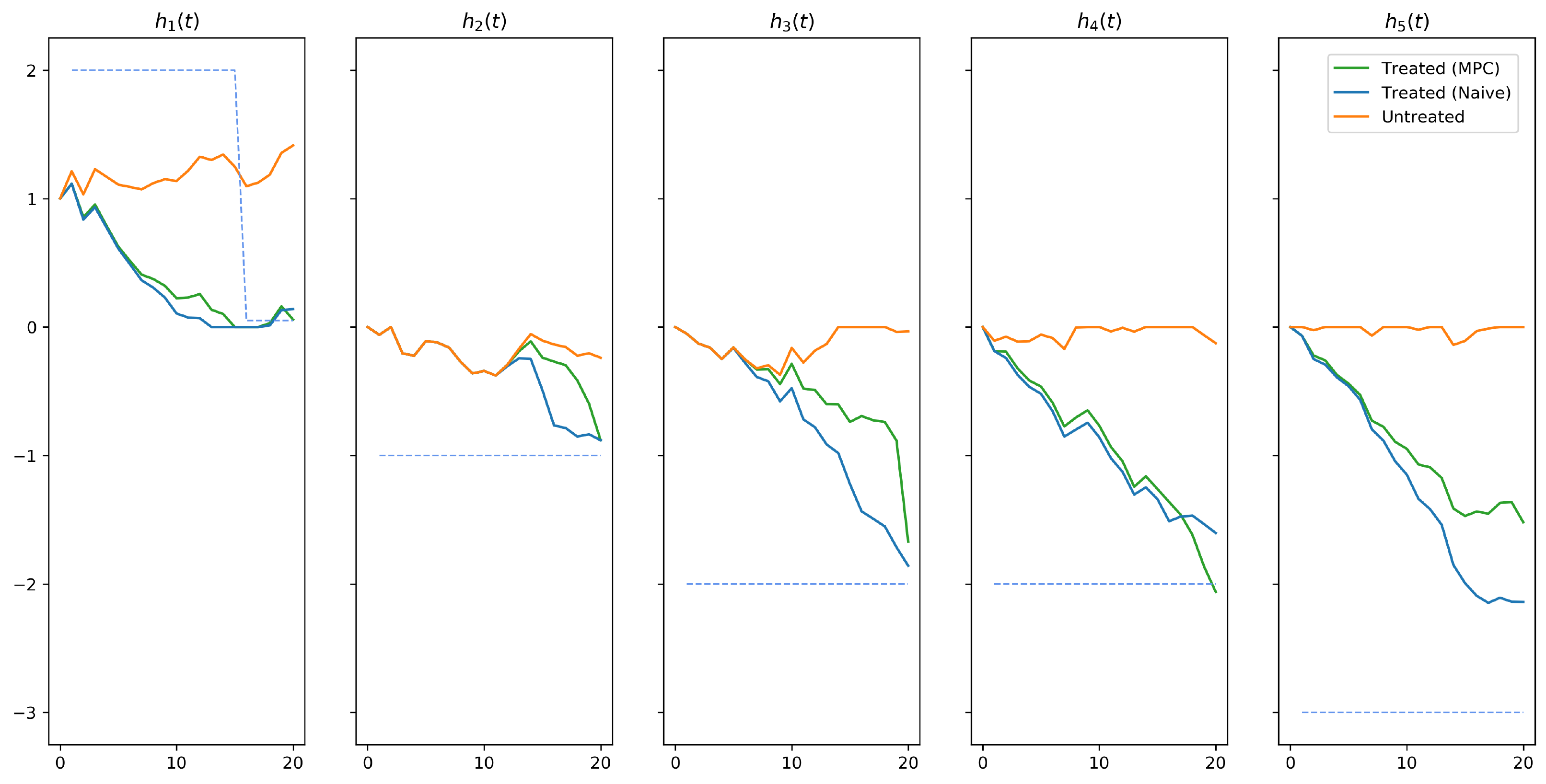}
		\end{subfigure}
		\caption{Optimal (a) radiation dose and (b) health status 
			trajectories for Example \ref{sec:ex_simple_mpc} using MPC (green) 
			and a naive planning approach (blue). Dashed lines represent the 
			bounds on the dose and health status ($D_{t,i}$ and $H_{t,i}$, 
			respectively). The MPC plan's health trajectories all remain within 
			the desired bounds, despite the error in the health dynamics model.}
		\label{fig:ex2_mpc_traj}
	\end{center}
\end{figure}

\clearpage

\section{Operator splitting}
\label{sec:op_split}
MPC enables us to robustly handle uncertainty over time. However, another 
challenge in radiation treatment planning is the sheer size of problems, 
which makes them computationally difficult to solve in practice. A typical 
case with $K = 15$ and $n = 10^4$ requires approximately $10^5$ 
floating-point operations for the beam-to-dose calculation alone. Over a 
month of sessions, that comes out to 4.5 million operations on a single 
machine.

In this section, we propose a fast, efficient method for solving the 
radiation treatment planning problem using operator splitting. Our method 
is distributed and scales readily with the number of beams as well as the 
length of treatment. It can be applied both to the original problem 
(\ref{prob:dyn_single}) and the soft constrained MPC variant 
(\ref{prob:dyn_single_mpc_slack}). Below, we describe the mathematical 
details for the former; the latter is a straightforward extension.

\subsection{Consensus form}
\label{sec:consensus}
We first rewrite problem (\ref{prob:dyn_single}) in an equivalent consensus 
form:
\BEQ
\label{prob:consensus}
\begin{array}{lll}
	\mbox{minimize} & \sum_{t=1}^T \phi_t(d_t) + 
		\sum_{t=1}^T \psi_t(h_t) & \\
	\mbox{subject to} & h_t = f_t(h_{t-1}, \tilde d_t), 
		\quad 0 \leq \tilde d_t \leq D_t, 
		\quad & t = 1,\ldots,T, \\
	& h_{t,i} \leq H_{t,i},~ i \in \mathcal{T}, 
		\quad h_{t,i} \geq H_{t,i},~ i \notin \mathcal{T}, 
		\quad & t = 1,\ldots,T, \\
	& d_t = A_tb_t, \quad 0 \leq d_t \leq D_t, 
		\quad 0 \leq b_t \leq B_t, \quad & t = 1,\ldots,T, \\
	& d_t = \tilde d_t, \quad \quad & t = 1,\ldots,T
\end{array}
\EEQ
with additional variable $\tilde d = (\tilde d_1,\ldots,\tilde d_T)$. 
This splits the problem into two parts, one that encapsulates the 
radiation physics and the other that contains the health dynamics. The 
parts share no variables. They are only linked by the consensus 
constraint, $d_t = \tilde d_t$, which requires their doses be equal.

\subsection{ADMM}
\label{sec:admm}
We solve problem (\ref{prob:consensus}) using an iterative algorithm 
called the alternating direction method of multipliers (ADMM) \cite{ADMM}.
In ADMM, the beams and health statuses are optimized separately, taking 
into account the difference between their resulting dose values. This 
difference is associated with a dual variable $u = (u_1,\ldots,u_T)$, 
where each $u_t \in \reals^K$, which is updated every iteration in order 
to promote consensus. 
\clearpage
\begin{algdesc}
	\label{algo:admm}
	\emph{ADMM algorithm.}
	\begin{tabbing}
		{\bf input}: initial point $(\tilde d^0, u^0)$, parameter $\rho > 0$.
			\\*[\smallskipamount]
		{\bf for} $k = 0,1,\ldots$ {\bf do} \\
		\qquad \= 1.\ \emph{Calculate beams.} For $t = 1,\ldots,T$, set 
			the value of $(b_t^{k+1}, d_t^{k+1})$ to a \\ 
		\> solution of the problem \\
		\> \qquad $\begin{array}{ll}
				\mbox{minimize} & \phi_t(d_t) + \frac{\rho}{2}\|d_t - 
					\tilde d_t^k - u_t^k\|_2^2 \\
				\mbox{subject to} & d_t = A_tb_t, \quad 0 \leq d_t \leq 
					D_t, \quad 0 \leq b_t \leq B_t.
			\end{array}$ \\
		\> 2.\ \emph{Calculate health trajectory.} Set the value of 
			$(h^{k+1}, \tilde d^{k+1})$ to a solution \\
		\> of the problem \\
		\> \qquad $\begin{array}{lll}
			\mbox{minimize} & \sum_{t=1}^T \psi_t(h_t) + \frac{\rho}{2}
				\|\tilde d - d^{k+1} + u^k\|_2^2 & \\
			\mbox{subject to} & h_t = f_t(h_{t-1}, \tilde d_t), 
				\quad 0 \leq \tilde d_t \leq D_t, 
				\quad & t = 1,\ldots,T, \\
			& h_{t,i} \leq H_{t,i},~ i \in \mathcal{T}, 
				\quad h_{t,i} \geq H_{t,i},~ i \notin \mathcal{T}, 
				\quad & t = 1,\ldots,T.
		\end{array}$ \\
		\> 3.\ \emph{Update dual variables.} $u^{k+1} := u^k + 
			\tilde d^{k+1} - d^{k+1}$. \\
		{\bf until} stopping criterion (\ref{cond:admm_stop}) is satisfied.
	\end{tabbing}
\end{algdesc}
Here $1/\rho> 0$ may be interpreted as the step size. Notice that the first 
step of Algorithm \ref{algo:admm} can be parallelized across sessions. We 
impose the dose bound constraint on both the beam and health subproblems 
because it produces faster convergence in practice.

\paragraph{Initialization.} For complex problems, the initial dose point 
$\tilde d^0$ can have a significant impact on the performance of Algorithm 
\ref{algo:admm}. Below, we describe one heuristic that produces a good 
starting point by solving a series of simple optimization problems. The 
idea is to first find the optimal beams in the static setting, where the 
entire treatment is delivered in a single session, then rescale these 
static beams each session to form a plan in the dynamic setting. By 
limiting ourselves to scaling factors, we significantly reduce the size of 
the dynamic problem.

We begin by solving the static treatment planning problem
\BEQ
\label{prob:init_static}
\begin{array}{ll}
	\mbox{minimize} & \phi_1(d_1) + \psi_1(h_1) + \mu \ones^T \zeta \\
	\mbox{subject to} & h_1 = f_1(h_0, d_1), \quad \zeta \geq 0, \\
	& h_{1,i} \leq H_{T,i},~ i \in \mathcal{T}, 
		\quad h_{1,i} \geq H_{T,i} - \zeta_i,~ i \notin \mathcal{T}, \\
	& d_1 = A_1b_1, \quad 0 \leq d_1 \leq \sum_{t=1}^T D_t, 
		\quad 0 \leq b_1 \leq \sum_{t=1}^T B_t
\end{array}
\EEQ
with respect to $b_1 \in \reals^n, d_1 \in \reals^K, h_1 \in \reals^K$, and 
$\zeta \in \reals^K$, where $\mu > 0$ is a slack penalty parameter. A 
reasonable choice for $\mu = \frac{1}{K - |\mathcal{T}|}$, assuming there 
is at least one non-target structure. Problem (\ref{prob:init_static}) is 
convex and can be easily handled on a single machine (\eg, via 
interior-point methods) for up to $10^5$ beams. Denote the optimal beam 
intensities by $b^{\text{stat}}$.

Next, we consider the dynamic treatment planning problem in which the 
beams for each session are restricted to be a scalar multiple of 
$b^{\text{stat}}$,
\BEQ
\label{prob:init_scale}
\begin{array}{lll}
	\mbox{minimize} & \sum_{t=1}^T \phi_t(d_t) + \sum_{t=1}^T \psi_t(h_t) 
		+ \mu \sum_{t=1}^T \ones^T\zeta_t & \\
	\mbox{subject to} & h_t = f_t(h_{t-1}, d_t), \quad \zeta_t \geq 0, 
		\quad & t = 1,\ldots,T, \\
	& h_{t,i} \leq H_{t,i},~ i \in \mathcal{T}, 
		\quad h_{t,i} \geq H_{t,i} - \zeta_{t,i},~ i \notin \mathcal{T}, 
		\quad & t = 1,\ldots,T, \\
	& d_t = \nu_tb^{\text{stat}}, \quad 0 \leq d_t \leq D_t, 
		\quad \nu_t \geq 0, \quad & t = 1,\ldots,T
\end{array}
\EEQ
with variables $(\nu_1,\ldots,\nu_T), (d_1,\ldots,d_T), (h_1,\ldots,h_T)$, 
and $(\zeta_1,\ldots,\zeta_T)$, where each $\nu_t \in \reals$ and $\zeta_t 
\in \reals^K$. This problem can be solved using a slight variation on 
Algorithm \ref{sec:ccp}. (For the initial CCP point, we may use the optimal 
time-invariant $\nu_t = \nu$ when $\beta_t = 0$; finding this value entails 
solving a small convex problem). Since there are only $O(TK)$ variables, 
convergence is generally quick, taking less than 5 iterations in our 
experiments. We use the resulting doses as our initial dose point for ADMM, 
\ie, $\tilde d_t^0 = \nu_t^{\star}b^{\text{stat}}$ for $t = 1,\ldots,T$.

Besides providing a good starting point, this initialization heuristic also 
gives us a way to quickly tune problem parameters. If the health trajectory 
from $\tilde d^0$ is poor, it is much faster to modify weights and re-solve 
problems (\ref{prob:init_static}) and (\ref{prob:init_scale}) than it is to 
re-run the full ADMM algorithm.

\paragraph{Stopping criterion.} If problem (\ref{prob:consensus}) is convex, 
then under mild conditions, ADMM converges to a solution assuming one exists. 
Moreover, the primal and dual residuals
\begin{align}
\label{eq:opt_residual}
	r_{\text{prim}}^k &= d^k - \tilde d^k \\
	r_{\text{dual}}^k &= \rho(\tilde d^k - \tilde d^{k-1})
\end{align}
also converge to zero. Thus, a reasonable stopping criterion is
\BEQ
\label{cond:admm_stop}
	\|r_{\text{prim}}^k\|_2 \leq \epsilon_{\text{prim}} \quad \mbox{and} 
		\quad \|r_{\text{dual}}^k\|_2 \leq \epsilon_{\text{dual}},
\EEQ
where $\epsilon_{\text{prim}} > 0$ and $\epsilon_{\text{dual}} > 0$ are 
tolerances for primal and dual feasibility, respectively. Typically, these 
tolerances are chosen with respect to absolute and relative cutoffs 
$\epsilon_{\text{abs}} > 0$ and $\epsilon_{\text{rel}} > 0$ using the 
relation
\begin{align*}
	\epsilon_{\text{prim}} &= \epsilon_{\text{abs}}\sqrt{TK} + 
		\epsilon_{\text{rel}}\max(\|d^k\|_2, \|\tilde d^k\|_2) \\
	\epsilon_{\text{dual}} &= \epsilon_{\text{abs}}\sqrt{TK} + 
		\epsilon_{\text{rel}}\|u^k\|_2.
\end{align*}
A common choice for $\epsilon_{\text{rel}} = 10^{-3}$, while the choice 
for $\epsilon_{\text{abs}}$ depends on the scale of the treatment planning 
problem \cite[\S 3.3.1]{ADMM}.

\paragraph{Convergence and choice of $\rho$.} When the problem is convex, 
\ie, the health dynamics function is affine, Algorithm \ref{algo:admm} 
converges to a solution for any $\rho > 0$, although the value of $\rho$ 
may have an impact on the practical convergence rate. When the problem is 
nonconvex, ADMM is a heuristic and the final beam/dose plan can depend 
directly on $\rho$ \cite[\S 9]{ADMM}. The question of how to choose $\rho$ 
is still unsettled; see \cite{GhadimiJohansson:2015,XuGoldstein:2017,XuLiuYang:2017} 
for further discussion on the topic. We have found that for data on the 
order of one, values of $\rho$ between $10^{-2}$ and $10^2$ work reasonably 
well.

\subsection{Clinical example}
\label{sec:ex_clinical}

\paragraph{Problem instance.} We test our method on a fluence map 
optimization of a prostate cancer IMRT case with $n = 34848$ beams and 
$K = 7$ structures consisting of a single PTV ($i = 1$), five OARs, and 
generic body voxels ($i = 7$). Treatment is carried out over $T = 45$ 
sessions, so the planning problem has about 1.6 million variables. The 
matrix $A_t$ remains constant over time and maps the beam intensities to 
the {\em average} dose per structure, \ie, $d_{t,i}$ is the total dose to 
structure $i$ divided by the number of voxels in $i$. Each beam's intensity 
cannot exceed $B_t = 0.025$.  

The LQ model parameters, initial health status, and dose and health status 
bounds can be found in Table \ref{table:prostate_parms}; these have been 
adapted from prior clinical datasets \cite{Kehwar:2005,GaoMayr:2010,MarksYorke:2010,LeeuwenOei:2018}. 
We choose the health and dose penalty functions to be 
\[
\psi_t(h_t) = (h_{t,1})_+ + \frac{1}{6}\sum_{i=2}^7 (h_{t,i})_-, \quad 
\phi_t(d_t) = \sum_{i=1}^6 d_{t,i}^2 + 0.25d_{t,7}^2, \quad t = 1,\ldots,T.
\]
These penalties place greater importance on reducing the health status of 
the PTV compared to sparing the OARs or generic body tissue. 

\begin{table}
	\caption{Prostate IMRT Problem Parameters}
	\label{table:prostate_parms}
	\centering
	\begin{adjustbox}{max width=\textwidth}
	\begin{tabular}{llcccccc} \toprule
	& & \multicolumn{3}{l}{LQ model} & \multicolumn{3}{l}{Health and dose} \\
	\cmidrule{3-8}
	$i$ & Structure & $\alpha_{t,i}$ & $\beta_{t,i}$ & $\gamma_{t,i}$ & $h_{0,i}$ & $H_{t,i}$ & $D_{t,i}$ \\
	\midrule
	1 & Prostate & 0.15 & 0.05 & $\begin{cases} 0 & t \leq 28 \\ 
		0.0173 & t > 28 \end{cases}$ & 5.8579 & $\begin{cases} 5.8579 & t \leq 14 \\ 
		4.4716 & 15 \leq t \leq 31 \\
		0 & t > 31 \end{cases}$ & 10 \\
	2 & Urethra  & 1    & 0.2  & 0 & 0 & -4.8 & 10 \\ 
	3 & Bladder  & 1    & 0.2  & 0 & 0 & -4.8 & 10 \\
	4 & Rectum   & 1    & 0.2  & 0 & 0 & -4.8 & 10 \\
	5 & L. Femoral Head & 1 & 0.25 & 0 & 0 & -3.0 & 10 \\
	6 & R. Femoral Head & 1 & 0.25 & 0 & 0 & -3.0 & 10 \\
	7 & Body     & 1    & 0.3333 & 0 & 0 & -6.0 & 10 \\
	\bottomrule
	\end{tabular}
	\end{adjustbox}
\end{table}

\paragraph{Computational details.} The computational setup is the same as 
in Example \ref{sec:ex_simple}. To solve the ADMM subproblems, we used 
MOSEK and ran CCP ($\lambda = 10^4$) on the health trajectory subproblem. 
With $\rho = 80$, ADMM converged in $82$ iterations to cutoffs of
$\epsilon_{\text{abs}} = 10^{-2}$ and $\epsilon_{\text{rel}} = 10^{-3}$. 
The normed residuals, $\|r_{\text{prim}}^k\|_2$ and 
$\|r_{\text{dual}}^k\|_2$, are shown in Figure \ref{fig:ex3_residuals}. 
Total runtime was about 43 minutes, with the bulk of that time spent on the 
main ADMM loop (initialization took only 32 seconds). By contrast, a 
straightforward application of Algorithm \ref{algo:seq_cvx} to this problem 
required over an hour.

\paragraph{Results and analysis.} Figure \ref{fig:ex3_traj_dose} depicts 
the dose trajectories resulting from the initial plan (green) and the final 
plan output by ADMM (blue). The initial plan is essentially a piecewise 
equal-dose fractionation scheme, reflected by the flat plateaus in the 
corresponding dose trajectories. This already gives us a good approximation 
of the final plan: both plans maintain a relatively high dose to the PTV 
of about 0.9 Gy until session 31, then drop off sharply to the same constant 
doses thereafter. However, during the high dose phase, the final plan 
gradually increases the dosage over time to all structures except the 
bladder ($i = 3$). By adapting dynamically to changes in the patient's 
anatomy, it is able to deliver more dose per session and thus achieve better 
tumor control, while still respecting the limits on the OARs' health statuses.

Indeed, we see in Figure \ref{fig:ex3_traj_health} that the final plan 
exactly attains the desired PTV health status of zero for $t > 31$. It must 
sacrifice some OARs to do this, reducing the health statuses of the urethra, 
rectum, and right femoral head ($i = 2,4$, and $6$) to their lower bounds, 
but never violates those bounds. In fact, by shifting radiation to other 
structures, the final plan actually improves the health of the bladder over 
that from the initial plan, which results in a $h_3(t)$ far below the limit 
of $-4.8$ for $t \geq 35$. Overall, it is clear that the combination of a 
solid initialization heuristic and ADMM produces a treatment plan that 
satisfies or even exceeds all of our clinical goals.

\begin{figure}
	\begin{center}
		\includegraphics[width=0.95\textwidth]{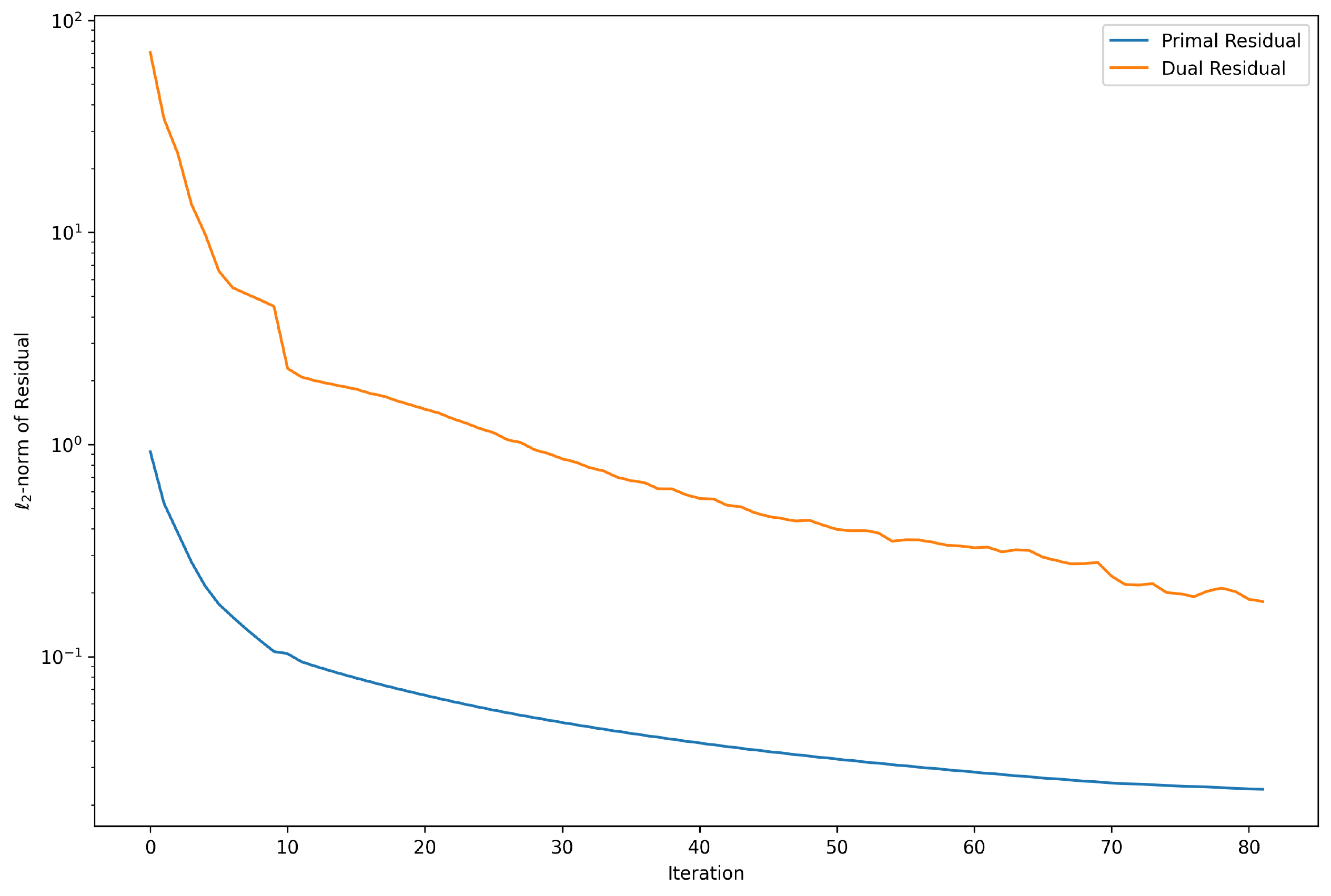}
	\end{center}
	\caption{Primal and dual residual $\ell_2$-norms for Example \ref{sec:ex_clinical}.}
	\label{fig:ex3_residuals}
\end{figure}

\begin{figure}
	\begin{center}
		\includegraphics[width=0.95\textwidth]{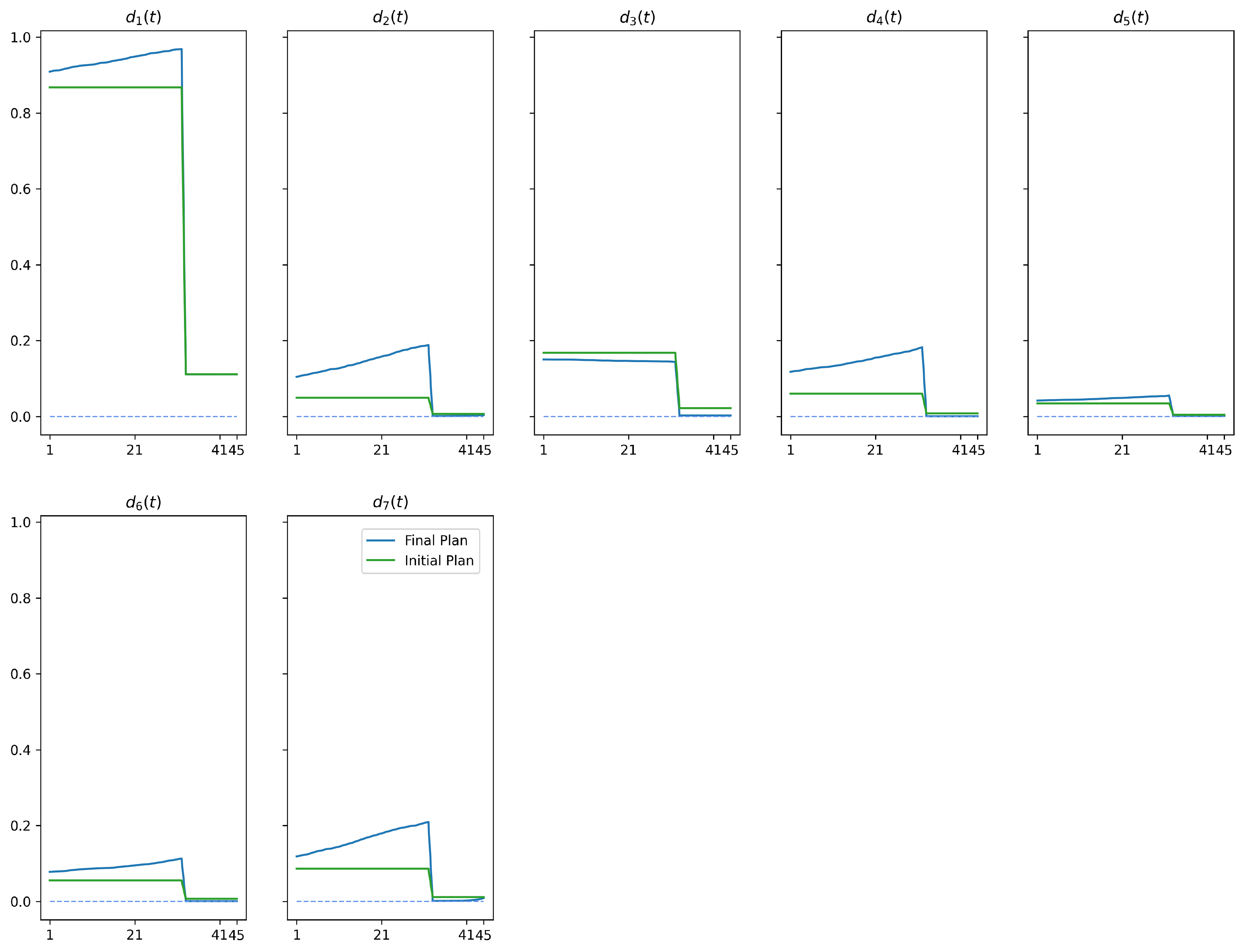}
	\end{center}
	\caption{Optimal radiation dose trajectory for Example \ref{sec:ex_clinical}. 
		The initial plan (green) depicts the dose output by the initialization 
		heuristic described in \S \ref{sec:admm}, while the final plan (blue) 
		depicts the dose output by the ADMM algorithm.}
	\label{fig:ex3_traj_dose}
\end{figure}

\begin{figure}
	\begin{center}
		\includegraphics[width=0.95\textwidth]{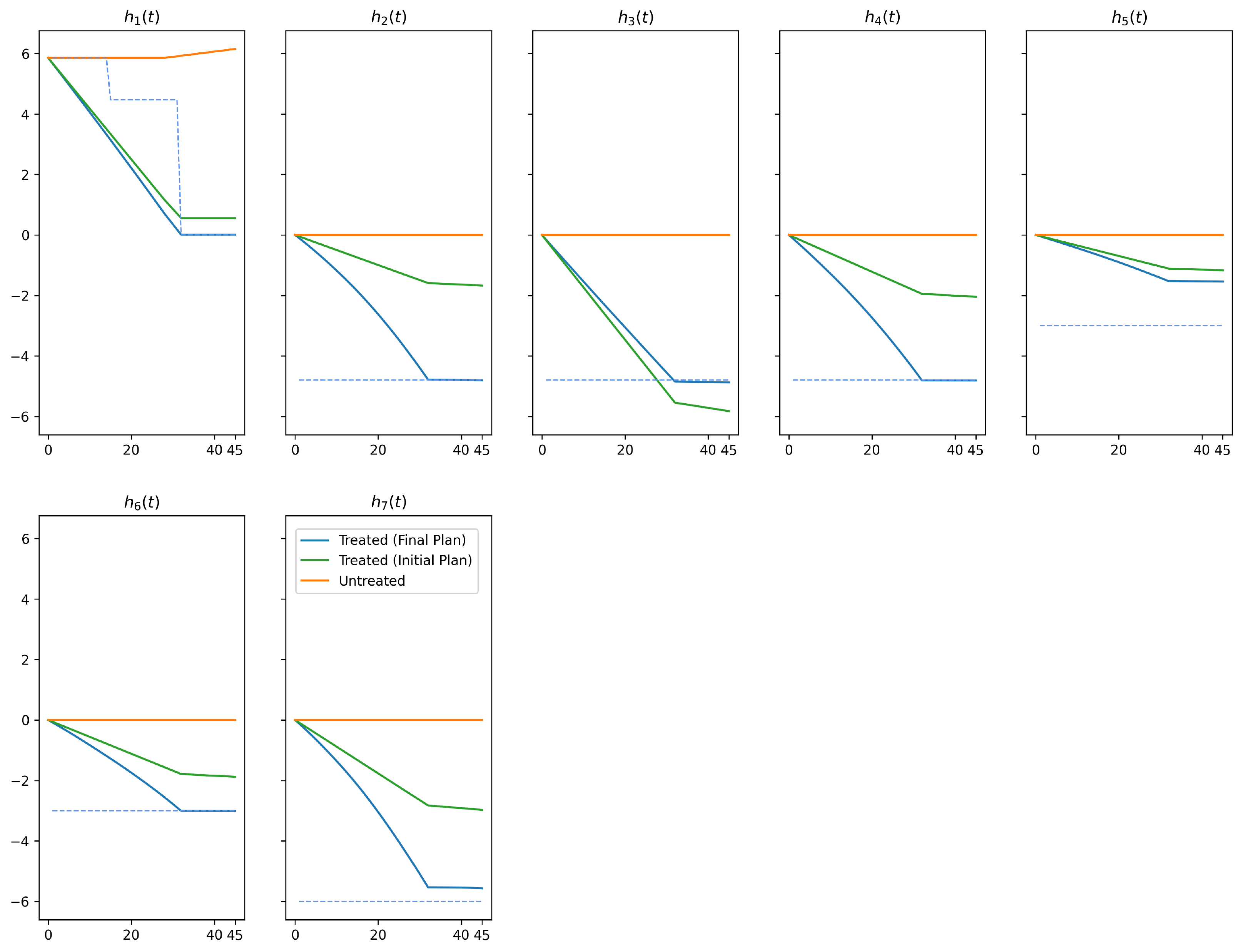}
	\end{center}
	\caption{Optimal health trajectories resulting from the doses in Figure 
		\ref{fig:ex3_traj_dose}. Dashed lines represent the health status bounds 
		$H_{t,i}$.}
	\label{fig:ex3_traj_health}
\end{figure}

\section{Implementation}
\label{sec:software}

We provide an implementation of our adaptive radiation treatment planning 
method in \adarad{}, an open-source Python software package based on CVXPY 
\cite{CVXPY}. Our implementation is fully distributed, leveraging Python's 
built-in multiprocessing library to execute solves in parallel. Users can 
quickly import patient data, define clinical goals, construct treatment 
plans, and visualize the results. They can also rapidly modify and re-plan 
a case, allowing for comparisons between different prescriptions and 
treatment lengths. Moreover, since \adarad{} is a Python library, it can 
be easily integrated with other libraries (\eg, for image processing) used 
in radiation therapy.

The code below imports some patient data and a prescription, solves for 
the optimal treatment plan, and plots the resulting dose and health 
trajectories.

\begin{verbatim}
import adarad, numpy
from adarad import Case, CasePlotter

# Construct the clinical case.
> case = Case()
> case.import_file("/examples/patient_01-case.yaml")
> case.physics.dose_matrix = numpy.load("/examples/patient_01-dmat.npy")

# Solve using ADMM algorithm.
> status, result = case.plan(slack_weight = 50, max_iter = 100, 
                             solver = ECOS, use_admm = True)
> print("Solve status: {}".format(status))
> print("Solve time: {}".format(result.solver_stats.solve_time))
> print("Iterations: {}".format(result.solver_stats.num_iters))

# Plot the dose and health trajectories.
> caseviz = CasePlotter(case)
> caseviz.plot_treatment(result, stepsize = 10)
> caseviz.plot_health(result, stepsize = 10)
\end{verbatim}
In this example, the dose matrix $A_t$ is the same for all $t$ and stored 
in a single \texttt{*.npy} file. \adarad{} also supports other sparse data 
representations, such as \texttt{scipy.csc\_matrix}. To specify a 
time-varying dose matrix, the user would input a list of matrices in order 
$[A_1, \ldots, A_T]$.

We start by constructing a \texttt{Case}, which contains \texttt{Anatomy}, 
\texttt{Physics}, and \texttt{Prescription} objects. The \texttt{Anatomy} 
and \texttt{Physics} must be defined prior to planning, either by manually 
specifying them in the code or importing a case description. A description 
is a YAML file that contains at minimum the keys \texttt{treatment\_length} 
and \texttt{structures}, where the latter is a list of anatomical structures 
$i = 1,\ldots,K$, each of which has a \texttt{name}, \texttt{is\_target} 
boolean indicator, and \texttt{alpha}, \texttt{beta}, and \texttt{gamma} 
values corresponding to the LQ model parameters. The initial health status 
and health and dose bounds may also be specified.

Once the \texttt{Case} is defined, we can solve for the optimal treatment 
plan. The \texttt{plan} function implements Algorithms \ref{algo:seq_cvx} 
and \ref{algo:admm} (the latter with \texttt{use\_admm = True}). It takes 
as optional input \texttt{d\_init}: the initial dose point, \texttt{use\_slack}: 
a boolean indicating whether to include slack variable $\delta$, 
\texttt{slack\_weight}: the slack penalty parameter $\lambda$, 
\texttt{max\_iter}: the maximum number of iterations, and \texttt{solver}: 
the convex solver to use for the beam and health subproblems. In the above 
example, we call the solver ECOS \cite{ECOS}, one of several free, 
open-source solvers packaged with CVXPY. If MOSEK is installed, we can call 
it as well by passing \texttt{solver = MOSEK} into the planning function.

After the algorithm finishes, \texttt{plan} saves the results in 
\texttt{case.current\_plan} and returns the final solve status along with 
a \texttt{RunRecord} object that carries solver performance data, such as 
the total runtime, and the optimal variable values. To visualize the 
resulting plan, we instantiate a \texttt{CasePlotter} object and call 
\texttt{plot\_treatment} and \texttt{plot\_health} on the \texttt{RunRecord} 
to display the dose and health trajectories, respectively. We can also 
extract the optimal beams, doses, and health statuses with, \eg, 
\texttt{result.beams} for further processing.

If we wish to explore alternate plans, we can easily modify the dose and 
health status constraints of any structure and re-plan the case. Re-planning 
is generally fast, since \adarad{} uses the previously stored solution as a 
warm start point. In a typical workflow, we may import a prescription formed 
from general clinical guidelines, then repeatedly adjust the dose/health 
status bounds until we obtain a treatment plan with our desired properties. 
The \texttt{case.current\_plan} will be updated with the new optimal values 
after each run. To keep a history of plans for comparison, we can save our 
results in the \texttt{Case} by calling \texttt{save\_plan} before 
re-optimizing. The code below provides an example of changing the upper 
dose bound on the PTV to $D_{t,i} = 10$ Gy for all sessions and plotting the 
dose and health trajectories under this new constraint alongside the 
trajectories of the original plan.
\begin{verbatim}
# Save previous treatment plan.
> case.save_plan("Original Plan")

# Constraint allows maximum of 10 Gy per session on the PTV.
> case.prescription["PTV"].dose_upper = 10

# Re-plan the case with new dose constraint.
> status2, result2 = case.plan(slack_weight = 50, max_iter = 100, 
                               solver = ECOS, use_admm = True)
> print("Solve status: {}".format(status2))

# Compare original and new treatment plans.
> caseviz.plot_treatment(result2, stepsize = 10, label = "New Plan", 
                         plot_saved = True)
> caseviz.plot_health(result2, stepsize = 10, label = "New Plan", 
                      plot_saved = True)
\end{verbatim} 
For more details on \adarad{}'s functions as well as additional examples, 
see the documentation at \url{https://github.com/anqif/adarad}.

\section{Conclusion}
\label{sec:conclusion}
To achieve the best outcomes, radiation therapy must adapt to new 
information about the patient's health and anatomy during treatment. 
We have described one method for adaptive radiation treatment planning 
using an operator splitting algorithm. Our method is highly scalable, 
parallelizable, and can efficiently handle a large number of beams and 
sessions. Moreover, it is robust to errors in the patient's health response 
model, as well as other sources of uncertainty in the clinic. We 
demonstrated its effectiveness on a large prostate cancer case and showed 
that the resulting plan improves markedly on a standard equal-dose 
fractionation scheme. 

Future work will focus on expanding our health response model to include 
sublethal damage repair, redistribution, and reoxygenation effects. We will 
also incorporate dose-volume constraints into the optimal control problem. 
Finally, to increase our algorithm's speed, we intend to release an 
implementation that takes advantage of the parallel processing capabilities 
of the GPU.

\section*{Acknowledgments}
We thank Peng Dong for providing the anonymized dataset for the prostate 
cancer IMRT case. This research was supported by the Stanford Graduate 
Fellowship.

\clearpage
\bibliographystyle{alpha}
\bibliography{adapt_rad_therapy}
\nocite{BoydVandenberghe:2004}

\end{document}